\documentclass[a4paper,UKenglish,cleveref, autoref, thm-restate]{lipics-v2021}

\hideLIPIcs  

\title{A faster algorithm for the Fréchet distance in 1D for the imbalanced case} 

\author{Lotte Blank}{University of Bonn, Germany }{lblank@uni-bonn.de}{}{}

\author{Anne Driemel}{University of Bonn, Germany}{driemel@cs.uni-bonn.de}{}{}

\authorrunning{L. Blank and A. Driemel} 

\Copyright{Lotte Blank and Anne Driemel} 



\category{} 

\relatedversion{} 

\nolinenumbers 

\newcommand{\F}{{Fréchet distance }}
\newcommand{\Fm}{{Fréchet distance}}
\usepackage[ruled, linesnumbered, vlined]{algorithm2e}
\usepackage{comment}
\usepackage{xcolor}

\EventEditors{John Q. Open and Joan R. Access}
\EventNoEds{2}
\EventLongTitle{42nd Conference on Very Important Topics (CVIT 2016)}
\EventShortTitle{CVIT 2016}
\EventAcronym{CVIT}
\EventYear{2016}
\EventDate{December 24--27, 2016}
\EventLocation{Little Whinging, United Kingdom}
\EventLogo{}
\SeriesVolume{42}
\ArticleNo{23}

\begin{document}
\maketitle
\begin{abstract}
The fine-grained complexity of computing the \F has been a topic of much recent work, starting with the quadratic SETH-based conditional lower bound by  Bringmann from 2014. Subsequent work established largely the same complexity lower bounds for the \F in 1D. However, the imbalanced case, which was shown by Bringmann to be tight in dimensions $d\geq 2$, was still left open.
Filling in this gap, we show that a faster algorithm for the \F in the imbalanced case is possible: Given two 1-dimensional curves of complexity $n$ and $n^{\alpha}$ for some $\alpha \in (0,1)$, 
we can compute their \F  in $O(n^{2\alpha} \log^2 n + n \log n)$ time. This rules out a conditional lower bound of the form $O((nm)^{1-\varepsilon})$ that Bringmann showed for $d \geq 2$ and any $\varepsilon>0$ in turn showing a strict separation with the setting $d=1$.
At the heart of our approach lies a data structure that stores a 1-dimensional curve $P$ of complexity $n$, and supports queries with a curve $Q$ of complexity~$m$ for the continuous \F between $P$ and $Q$. The data structure has size in $\mathcal{O}(n\log n)$ and uses query time in $\mathcal{O}(m^2 \log^2 n)$. 
Our proof uses a key lemma that is based on the concept of visiting orders and may be of independent interest. We demonstrate this by substantially simplifying the correctness proof of a clustering algorithm by Driemel, Krivo{\v s}ija and Sohler from 2015.

\end{abstract}
\newpage
\section{Introduction}
Since its introduction to Computational Geometry by Alt and Godau~\cite{AG95}, the  complexity of computing or approximating the \emph{\F} has been the topic of much research and debate~\cite{AAK2014,BBMM2017}. Alt conjectured in 2009~\cite{alt2009computational} that the associated decision problem may be 3-SUM hard, which would likely rule out subquadratic algorithms~\cite{gajentaan1995class}. With the advent of fine-grained complexity theory came the celebrated 
result by Bringmann~\cite{bringmann2014walking} in 2014, and we now know that strictly subquadratic algorithms with running time in $O(n^{2-\varepsilon})$ are indeed not possible, unless the Strong Exponential Time Hypothesis (SETH) fails. While Bringmann's construction uses curves in the plane, subsequent work by Bringmann and Mulzer~\cite{bringmann2016approximability} and later Buchin, Ophelders and Speckmann~\cite{BOS18} extended the result showing the same hardness also for curves in one dimension. Bringmann in his initial paper showed the tightness of the algorithm by Alt and Godau up to logarithmic factors also in the imbalanced case. Namely, given two polygonal curves in the plane of complexity $n$ and $m=n^{\alpha}$ for some $\alpha \in (0,1]$, there is no algorithm that computes the \F in time $O((nm)^{1-\varepsilon})$ for any $\varepsilon >0$ unless SETH fails. However, the imbalanced case was left open for curves in 1D in  \cite{bringmann2016approximability,BOS18}. 

In this paper, we fill this gap by showing a faster algorithm for computing the \F in one dimension ruling out lower bounds of the form $O((nm)^{1-\varepsilon})$ and thus showing a strict separation between the complexity in 1D and higher dimensions. 

\subparagraph*{Signatures and visiting orders}
In order to obtain our result, we give necessary and sufficient conditions for when the decision algorithm should accept a 1D instance. Our characterization uses the concept of \emph{signatures} which were earlier introduced by Driemel, Krivo{\v s}ija and Sohler~\cite{DKS16} in 2016. Intuitively, a $\delta$-signature of~$P$ is a curve that consists of important minima and maxima of $P$ preserving its overall shape. We combine this with the concept of  \emph{visiting orders} formalized by Bringmann, Driemel, Nusser and Psarros \cite{BDNP21}.
We show that the \F between two 1-dimensional curves is at most~$\delta$ if and only if there exist a so-called \emph{coupled $\delta$-visiting order}, subject to additional technical conditions. A coupled $\delta$-visiting order is a sequence of ordered tuples of important indices of the vertices that define the polygonal curves $P$ and $Q$ including the indices of the \emph{$\delta$-signature vertices}. 
Our  lemma leads to a shorter alternative proof of Theorem~3.7 of Driemel, Krivo{\v s}ija and Sohler~\cite{DKS16}, thereby substantially simplifying their proof of 17 pages in~\cite{DKS15}. Our proof makes use of a greedy algorithm devised by Bringmann et. al.~\cite{BDNP21}, which we slightly adapt to suit our needs.


\subparagraph*{\F  oracles}
We develop our results in the form of a \F  oracle and obtain a more refined result. A \F  oracle is a data structure that preprocesses an input curve to answer queries with a query curve for the \F between input and query curve. We give a brief history of results focussing on the continuous \Fm. 
\F oracles were introduced by Driemel and Har-Peled \cite{DH13} in 2012 in the context of their work on the shortcut \F in  fixed dimension $d$. They constructed a $(1+\varepsilon)$-approximate \F oracle for the case where the query curves are line segments. Their data structure uses space in~$\mathcal{O}(n(1/\varepsilon)^{2d}\cdot \log^2(1/\varepsilon))$ and query time in~$\mathcal{O}((1/\varepsilon)^2\log n\log\log n)$. It can also be used to answer queries to a sub-curve of $P$. They extended this data structure to support  queries of complexity $m$. In the extended version it has size~$\mathcal{O}(n\log n)$, uses query time in~$\mathcal{O}(m^2\log n\log(m\log n))$ and approximates the continuous \F up to a (large) constant factor. 
For exact \F queries, the picture looks much different. In a very recent result, Cheng and Huang~\cite{CH24} presented a data structure that builds upon point location in arrangements of polynomials of bounded degree. Their data structure has size in $\mathcal{O}(nm)^{\text{poly}(d,m)}$, expected query time in $\mathcal{O}((md)^{\mathcal{O}(1)}\log (nm))$ solving \F queries in $d$ dimensions, where the query complexity $m\geq 2$ is given at preprocessing time. 
Given the apparent difficulty of the problem, some works have focused on the restricted setting where the queries are line segments in the plane (i.e., $d=2$ and $m=2$) \cite{BMO17,BHOSSS22}. In this vein, Buchin, van der Hoog, Ophelders, Schlipf, Silveira, and Staals~\cite{BHOSSS22} 
describe a data structure that uses storage in $\mathcal{O}(nk^{3+\varepsilon}+n^2)$ and query time in $\mathcal{O}((n/k)\log^2 n+\log^4 n)$, where $k\in[1, \dotso, n]$ is a tradeoff parameter and $\varepsilon>0$ is an arbitrarily small constant.  
Another line of work studies approximate distance oracles in the discrete setting where the discrete \F is derived either from the Euclidean metric or from a general graph metric~\cite{AFKR24, driemel2019sublinear, DriemelHR22, FF23, HRW-2022}.


\subparagraph*{Distance approximation}
There has been an ongoing search for faster approximation algorithms~\cite{bringmann2016approximability, CR18, HO24, HKOS23}, however the gap between upper and lower bounds remains large. While the lower bounds do not rule out a subquadratic-time $3$-approximation, the best-known approximation factors that can be obtained in this regime remain polynomial in~$n$, even in 1D. Van der Horst and Ophelders show, given two 1-dimensional curves and a parameter $\alpha \in [1,n]$, one can $\alpha$-approximate the continuous \F  in time in $O(n \log^3 n + (n^2/{\alpha^3}) \log^2 n \log \log n)$~\cite{HO24} and this is currently the best possible in 1D.
For specialized classes of well-behaved curves better approximations are known~\cite{aronov2006frechet, buchin2017folding, DHW12, GMMW19}. 

\subparagraph*{Our Results}
Our results are for the exact computation of the continuous 
\F of 1-dimensional curves. Our data structure supports query curves of arbitrary complexity~$m$, which need not be known during preprocessing. We can preprocess a curve of complexity $n$ using storage and preprocessing time in $\mathcal{O}(n\log n)$ to support queries with time in $\mathcal{O}(m^{2}\log^2 n)$. Our data structure  leads to several improvements for the exact computation of the \F in 1D: First, it leads to an improved running time for the decision problem of the \F in the case that the complexity of the $\delta$-signatures is low. Specifically, given two curves of complexity $n$ and $m$, with $n\geq m$, we obtain an exact algorithm with running time in $\mathcal{O}(s_P s_Q\log n+n\log n)$, where $s_P$ and $s_Q$ denote the complexities of the $\delta$-signatures of the two curves. Second, if $m=n^\alpha$ for some $\alpha \in (0,1)$, we can compute the continuous \F in $\mathcal{O}(n^{2\alpha}\log^{2} n+n\log n)$ time. This  improves upon the running time of $\mathcal{O}(n^{1+\alpha}\log n)$ of the classical algorithm~\cite{AG95} and rules out lower bounds of type shown in~\cite{bringmann2014walking} for the imbalanced case.


\section{Preliminaries}
For any two points $p_1, p_2\in \mathbb{R}$, $\overline{p_1 p_2}$ denotes the directed line segment connecting $p_1$ with~$p_2$. A time series $P$ of complexity $n$ is formed by ordered line segments $\overline{P(i) P(i+1)}$ of points $P(1), P(2), \dotso, P(n)$, where $P(i)\in \mathbb{R}$. We obtain a polygonal curve that can be viewed as a function $P:[1, n]\rightarrow \mathbb{R}$, where $P(i+\alpha)=(1-\alpha)P(i)+\alpha P(i+1)$ for $i\in \{1, \dotso, n\}$ and $\alpha\in[0,1]$. This curve is also denoted with $\langle P(1), \dotso, P(n)\rangle$. We call the points $P(i)$ vertices and the line segments $\overline{P(i) P(i+1)}$ edges of $P$. Further, $P[s_1, s_2]$ is the subcurve of~$P$ starting in $P(s_1)$ and ending in $P(s_2)$ for $s_1\leq s_2$. We define $\min(P[s_1, s_2])=\min\{P(s)|\ s\in [s_1, s_2]\}$ and $\arg\min(P[s_1, s_2])$ to be an $s_1\leq s\leq s_2$ such that $P(s)=\min(P[s_1, s_2])$. The $\delta$-range of a time series $P$ is the interval $B(P, \delta)=\{x|\ \exists s\in [1,n] \text{ s.t. } |x-P(s)|\leq \delta\} $. 

Let $P: [1,n]\rightarrow \mathbb{R}$ and ${Q:[1,m]\rightarrow \mathbb{R}}$ be two time series. Then, the (continuous) \emph{\F} between them is defined as
\[ d_{\text{F}}(P, Q)=\min_{h_P\in \mathcal{F}_P, h_Q\in \mathcal{F}_Q }\ \max_{a\in [0,1]} | P(h_P(a))-Q(h_Q(a))|,\]
where $\mathcal{F}_P$ is the set of all continuous, non-decreasing functions $h_P: [0,1]\rightarrow [1,n]$ with ${h_P(0)=1}$ and $h_P(1)=n$, respectively $\mathcal{F}_Q$ for $Q$. We say that a point $Q(t)$ is \emph{matchable} to a point $P(s)$ for a value $\delta\geq 0$ if there exist functions $h_P\in \mathcal{F}_P$ and $h_Q\in \mathcal{F}_Q$ and a $b\in [0, 1]$ such that $\max_{a\in [0,1]} | P(h_P(a))-Q(h_Q(a))|\leq \delta$ and $h_P(b)=s$ and $h_Q(b)=t$. 

\begin{observation}\label{o:concatanation}
    Let $s\leq n$ and $t\leq m$. If $d_F(P[1, s], Q[1, t]), d_F(P[s,n], Q[t,m])\leq\delta$, then $d_F(P[1,n],Q[1,m])\leq \delta$.
\end{observation}

The next lemma is a generalization of Lemma~37 of Bringmann, Driemel, Nusser and Psarros \cite{BDNP21} and will be used to prove \cref{t:coupled visiting order}. We need the definition of a $\delta$-monotone time series. Let $P:[1, n]\rightarrow \mathbb{R}$ be a time series, then it is $\delta$-monotone increasing (resp. decreasing) if for all $s<s'\in [1, n]$, it holds that $P(s)\leq P(s')+\delta$ (resp. $P(s)\geq P(s')-\delta$).
\begin{lemma}[Generalization of Lemma~37 in \cite{BDNP21}] \label{l:VarLemma37}
    For two time series $P=\langle P(1), \dotso, P(n)\rangle$ and ${Q=\langle Q(1), \dotso Q(m)\rangle}$, it holds that $d_F(P,Q)\leq \delta$ if
    \begin{enumerate}[(i)]
        \item $P, Q$ are $2\delta$-monotone increasing (resp. decreasing),
        \item $|P(1)-Q(1)|\leq \delta$, $|P(n)-Q(m)|\leq \delta$,
        \item $P\subseteq B(Q, \delta)$, $Q\subseteq B(P, \delta)$, 
        \item $P(1)$ or $Q(1)$ is a global minimum (resp. maximum) of its time series, and
        \item $P(n)$ or $Q(m)$ is a global maximum (resp. minimum) of its time series.
    \end{enumerate}
\end{lemma}
The lemma can be proven using the same arguments as the proof of Lemma~37 in \cite{BDNP21}. For the sake of completeness, we provide a full proof in the appendix.

\subsection{Signatures and visiting orders}

Signatures and (coupled) visiting orders are key concepts used in this paper. 
\Cref{fig:visiting-order} gives an example of a $\delta$-signature.

\begin{definition}[$\delta$-signature]
    Let $P=\langle P(1), \dotso, P(n)\rangle$ be a time series and $\delta\geq 0$.
    Then, a \emph{$\delta$-signature} $P'=\langle P(i_1), \dotso, P(i_t)\rangle$ with $1=i_1<\cdots <i_t=n$ of $P$ is a time series with the following properties:
    \begin{enumerate}[(a)]
        \item \emph{(non-degenerate)} For $k=2, \dotso, t-1$, it holds $P(i_k)\notin \overline{P(i_{k-1}) P(i_{k+1})}$.
        \item \emph{($2\delta$-monotone)} For $k=1, \dotso, t-1$, $P[i_k, i_{k+1}]$ is $2\delta$-monotone increasing or decreasing.
        \item \emph{(minimum edge length)} If $t>2$, then for $k=2, \dotso, t-2$,  $|P(i_k)-P(i_{k+1})|>2\delta$, $|P(i_1)-P(i_2)|>\delta$, and $|P(i_{t-1})-P(i_t)|>\delta$.
        \item \emph{(range)} 
             For $k=2, \dotso, t-2$, it holds that $P(s)\in \overline{ P(i_{k}) P(i_{k+1})}$ for all $s\in [i_{k}, i_{k+1}]$,\\
            if $t>2$, then it holds $P(s)\in B(P(1), \delta)\cup \overline{P(1) P(i_2)}$ for all $s\in [1, i_2]$,\\
             if $t>2$, then it holds $P(s)\in B(P(n), \delta)\cup \overline{P(i_{t-1}) P(n)}$ for all $s\in [i_{t-1}, n]$, and\\
            if $t=2$, then $P(s)\in B(\overline{P(1)P(n)}, \delta)$ for all $s\in [1, n]$.
    \end{enumerate}
\end{definition}
Driemel, Krivo{\v s}ija and Sohler \cite{DKS16} showed that for every $\delta$ there exists a unique $\delta$-signature. 
We define an \emph{extended $\delta$-signature} $P''$ in the following way: If $t=2$, then $P''$ consists of $P(1)$ and $P(n)$ and the global minimum and global maximum of $P$ in order of their indices. Otherwise, we use $P'$ as defined above and we add two vertices with indices $1\leq i'< i_2$ and $i_{t-1}< i''\leq n$ such that $P[1, i_2]\subset \overline{P(i') P(i_2)}$ and $P[i_{t-1}, n]\subset \overline{P(i_{t-1}) P(i'')}$. The vertices of the extended $\delta$-signature $P''$ are called \emph{$\delta$-signature vertices} of $P$. Further, we say a point $P(s)$ of $P$ is \emph{supported by the $\delta$-signature edge} $\overline{P(i_k) P(i_{k+1})}$ if $i_k\leq s\leq i_{k+1}$. It holds that the subcurve of $P$ from the first until the second (resp. from the second last until the last) $\delta$-signature vertex of $P$ is contained in a $2\delta$-range.

The signatures have a unique hierarchical structure as the following lemma shows.
\begin{lemma}[Driemel, Krivo{\v s}ija and Sohler \cite{DKS16}]\label{l:hierachicalSignature}
Given a time series $P:[1, n]\rightarrow \mathbb{R}$ with vertices in general position, there exists a series of signatures $P_1, \dotso, P_k$ and corresponding parameters $0=\delta_1 <\delta_2<\cdots<\delta_{k+1}$, such that
\begin{enumerate}[(a)]
    \item $P_i$ is a $\delta$-signature of $P$ for any $\delta\in [\delta_i, \delta_{i+1})$,
    \item the vertex set of $P_{i+1}$ is a subset of the vertex set of $P_i$,
    \item $P_k$ is the linear interpolation of $P(1)$ and $P(n)$.
\end{enumerate}
\end{lemma}

With this hierarchical structure, it is possible to construct the following data structure.
\begin{theorem}[Driemel, Krivo{\v s}ija and Sohler \cite{DKS16}]\label{t:DS_signatures}
    Given a time series $P:[1, n]\rightarrow \mathbb{R}$ with vertices in general position, we can construct a data structure in time in $\mathcal{O}(n \log n)$ and space in $\mathcal{O}(n)$, such that given a parameter $l$ we can extract in time in $\mathcal{O}(l\log l)$ a signature of maximal size $l'$ with $l'\leq l$.
\end{theorem}
If we store in addition the points $(1, P(1)), (2, P(2)), \dotso, (n, P(n))$ in the data structure of \cref{t:successorProblem}, it is possible to compute the extended signature of $P$ of maximal complexity $l'$ with $l'\leq l$ in $\mathcal{O}(l\log l+\log n)$ time.

\begin{figure}
    \centering
    \includegraphics[page=3]{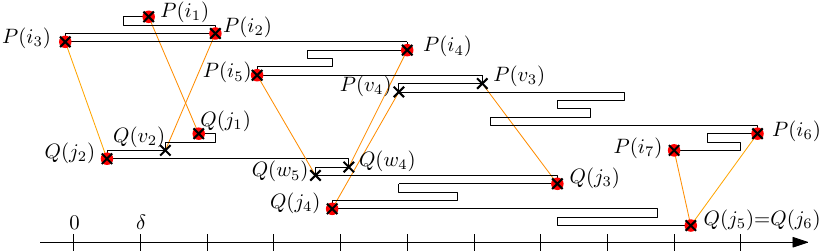}
    \caption{In this paper, the vertices of the time series are drawn as vertical segments for clarity. The $\delta$-signature vertices are marked with a red disk and $((i_1, j_1), (i_2, v_2),(i_3, j_2), (i_4, w_4),$ $ (i_5, w_5), (v_3, j_3), (v_4, j_4), (i_6, j_5), (i_7, j_6))$ is a coupled $\delta$-visiting order.}
    \label{fig:visiting-order}
\end{figure}

We adopt the concept of $\delta$-visiting orders as defined by Bringmann et. al. \cite{BDNP21}. A $\delta$-visiting order of $Q$ on $P$ for vertices $Q(j_1), \dotso, Q(j_t)$ consists of indices $v_1\leq\cdots\leq v_t$ of vertices of~$P$ such that $|P(v_k)-Q(j_k)|\leq \delta$ for all $k$. We extend this notion as follows. Our idea is to fix vertices of both time series and to interleave two $\delta$-visiting orders---one of $Q$ on $P$ and one of $P$ on $Q$. 
\Cref{fig:visiting-order} shows an example of a coupled $\delta$-visiting order, which we define further below.


\begin{definition}[coupled $\delta$-visiting order]
    Consider two time series $P=\langle P(1), \dotso, P(n)\rangle$ and $Q=\langle Q(1), \dotso, Q(m)\rangle$. Let $v_1\leq\cdots \leq v_{s_Q}$ be a $\delta$-visiting order of $Q$ on $P$ for the $\delta$-signature vertices $Q(j_1), \dotso, Q(j_{s_Q})$ and $w_1\leq\cdots \leq w_{s_P}$ be a $\delta$-visiting order of $P$ on~$Q$ for the $\delta$-signature vertices $P(i_1), \dotso, P(i_{s_P})$. 
    These two $\delta$-visiting orders are said to be crossing-free if there exists no $k, l$ such that $i_k< v_l$ and $j_l<w_k$, or $v_k< i_l$ and $w_l<j_k$. In this case, the ordered sequence containing all tuples $(v_k, j_k)$ and $(i_l, w_l)$, where $k=1, \dotso, s_Q$ and $l=1, \dotso, s_P$, is called \emph{coupled $\delta$-visiting order}.
\end{definition}

For the sake of a clean presentation, we allow to add a vertex to the time series $P$ that lies on one of the edges between the first and the second $\delta$-signature vertex of $P$ as well as one vertex between the second to last and the last $\delta$-signature vertex of $P$, resp. for $Q$ (see $Q(v_2)$ in \cref{fig:visiting-order}).  Adding those at most four vertices does not change the \Fm, it only helps us defining the $\delta$-visiting orders via vertex indices. Our algorithm described in Section~\ref{s:StBD} chooses the parameter for these vertices (named $v_{1,2}$, $w_{2,1}$, $\tilde{v}$, and $\tilde{w}$) when handling the traversals at the beginning and the ending parts of the two curves. In this context, we call these vertices \emph{pseudo-start}, resp., \emph{pseudo-end} \emph{vertex}. It is important to note that they are chosen with respect to the other curve.

\subsection{Orthogonal Range Successor Problem}

The \emph{Orthogonal range successor problem} is a variant of orthogonal range reporting. Here, a set $S$ of points must be stored such that the point in $S$ with the smallest $x$-coordinate (resp. smallest $y$-coordinate) contained in a query range $Q$ can be reported efficiently. The query range is an axis-aligned hyper-rectangle. For points in the plane, the \emph{layered range tree} is a modified version of range trees which uses fractional cascading \cite{CG86}, and can be used to solve the orthogonal range successor problem. For a comprehensive overview and further details, see Agarwal \cite{A04}. 

\begin{theorem}[\cite{A04, CG86}] \label{t:successorProblem}
In the pointer machine model, there exist a data structure solving the orthogonal range successor problem in the plane of size and preprocessing time in $\mathcal{O}(n\log n)$ that uses query time in $\mathcal{O}(\log n)$. In addition, it can return the number $m$ of stored points contained in an axis-aligned query rectangle in time in $\mathcal{O}(\log n)$ and those $m$ points can be returned in $\mathcal{O}(\log n+m)$ time.
\end{theorem}

\section{Key Lemma}
Driemel, Krivo{\v s}ija  and Sohler \cite{DKS16} showed that if $d_F(P,Q)\leq \delta$, then there exists a $\delta$-visiting order for the $\delta$-signature vertices of $P$ on $Q$. We strengthen their statement and show an equivalence that characterizes the \F in one dimension using the coupled $\delta$-visiting order. In the proof, we use the following lemma, which shows that if there exists any coupled $\delta$-visiting order, then there exists one of a specific form.

\begin{lemma}\label{l:visiting_order_wproperties}
    If there exists a coupled $\delta$-visiting order $((v_1, w_1), \dotso, (v_t, w_t))$ of $P$ and $Q$ with $d_F(P[1, v_2], Q[1, w_2])\leq \delta$ and $d_F(P[v_{t-1}, n], Q[w_{t-1}, m])\leq \delta$, then there exists a coupled $\delta$-visiting order $((v'_1, w'_1), \dotso, (v'_{t'}, w'_{t'}))$ with
    \begin{enumerate}[(a)]
        \item $d_F(P[1, v'_2], Q[1, w'_2])\leq \delta$ and $d_F(P[v'_{t'-1}, n], Q[w'_{t'-1}, m])\leq \delta$,
        \item $v'_2<v'_3<\cdots<v'_{t'-1}$ and $w'_2<w'_3<\cdots< w'_{t'-1}$, 
        \item $|P(v'_k)-Q(w'_{k+1})|>\delta$, $|P(v'_{k+1})-Q(w'_k)|>\delta$ for $k=2, \dotso, t'-2$, and
        \item $P[v'_k, v'_{k+1}]$ and $Q[w'_k, w'_{k+1}]$ are both $2\delta$-monotone increasing or both $2\delta$-monotone decreasing for $k=1, \dotso, t'-1$.
    \end{enumerate}
\end{lemma}
\begin{proof}
	If $w_2=w_3$, there exist values $w'\leq w_2$ and $v'\leq v_2$  such that $P[1, v_3]\subset B(Q(w'), \delta)$, $Q[w', w_3]\subset B(P(v_3), \delta)$ and ${d_F(P[1, v'], Q[1, w'])\leq\delta}$. Hence, $d_F(P[1, v_3], Q[1, w_3])\leq \delta$. The same holds in the case $v_2=v_3$, resp. for $v_{t-2}=v_{t-1}$ and $w_{t-2}=w_{t-3}$.
    If $v_k=v_{k+1}$ for a $k=2, \dotso, t-2$, then $Q(w_k)$ or $Q(w_{k+1})$ is no $\delta$-signature vertex. 
    Hence, we can omit $(v_k, w_k)$ or $(v_{k+1}, w_{k+1})$ and keep a coupled $\delta$-visiting order with Property~(b). If $|P(v_k)-Q(w_{k+1})|\leq \delta$ and $P(v_k)$ is a $\delta$-signature vertex, then $P(v_{k+1})$ and $Q(w_{k})$ are no $\delta$-signature vertices. In this case, we set $w_k=w_{k+1}$ and remove the tuple $(v_{k+1}, w_{k+1})$. Similarly, if $|P(v_{k+1})-Q(w_{k})|\leq \delta$.

    Consider the case that $Q(w_k)$ is a $\delta$-signature vertex, $Q[w_k, w_{k+1}]$ is $2\delta$-monotone increasing and $P[v_k, v_{k+1}]$ is not for a $k=1, \dotso, t$. Then, $P(v_k)$ cannot be $\delta$-signature vertex. Set $v'_k=\arg\min(P[v_k, v_{k+1}])$. Then, $P[v'_k, v_{k+1}]\subseteq[P(v'_k), P(v'_k)+2\delta]$ since $P[v_k, v_{k+1}]$ is $2\delta$-monotone decreasing. Hence, $P[v'_k, v_{k+1}]$ is also $2\delta$-monotone increasing. If $Q(w_{k+1})$ is a $\delta$-signature vertex, then $|P(v'_k)-Q(w_k)|\leq \delta$ since $Q(w_k)\leq Q(w_{k+1})-2\delta$ and $|Q(w_{k+1})-P(v_{k+1})|\leq \delta$. Otherwise, $P(v_{k+1})$ is a $\delta$-signature vertex and $v_{k+1}=v'_k$. Since $|Q(w_{k+1})-P(v_{k+1})|\leq \delta$, $|Q(w_{k})-P(v_{k})|\leq \delta$ and $Q(w_{k})$, $P(v_{k+1})$ are local minima, it follows that $|P(v'_k)-Q(w_k)|\leq \delta$. Hence, we can reassign $v_k=v'_k$. If $k=2$, we reassign $v_3$ instead of $v_2$ in a similar way. This does not interfere with the other reassignments by the monotonicity conditions of $\delta$-signatures.
    After the reassignment $((v_1, w_1), \dotso, (v_t, w_t))$ is a coupled\linebreak{$\delta$-visiting} order with Properties (a)-(d). 
\end{proof}
\begin{figure}
    \centering
    \includegraphics{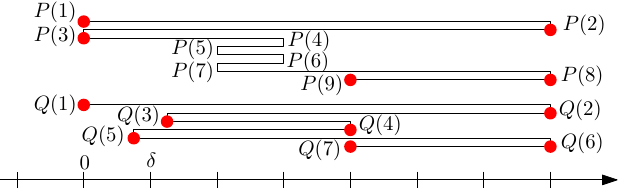}
    \caption{The red disks are the $\delta$-signature vertices. 
    The indices $1\leq 2\leq 5\leq 6\leq 7$ of $Q$ are a $\delta$-visiting order of the $\delta$-signature vertices of $P$ on $Q$, resp. $1\leq 2\leq 5\leq 6\leq 7\leq 8\leq 9$ of the $\delta$-signature vertices of $Q$ on $P$. Those are the only existing $\delta$-visiting orders for the $\delta$-signatures.
    Therefore, there does not exist a coupled $\delta$-visiting order of $P$ and $Q$, since $(3, 5)$ and $(5, 3)$ cross.}
    \label{fig:Counterexample visiting order}
\end{figure}
\begin{lemma}[Key Lemma] \label{t:coupled visiting order}
    Let $P, Q$ be two time series which are endowed with pseudo-start and pseudo-end vertices. Then, $d_F(P, Q)\leq \delta$ if and only if there exist a coupled $\delta$-visiting order $((v_1, w_1), \dotso, (v_t, w_t))$ of $P$ and $Q$ such that $d_F(P[1, v_2], Q[1, w_2])\leq \delta$ and $d_F(P[v_{t-1}, n], Q[w_{t-1}, m])\leq \delta$.
\end{lemma}

\Cref{fig:Counterexample visiting order} shows that it is not sufficient to use the notion of a $\delta$-visiting order alone. In this example, it holds that $d_F(P,Q)>\delta$ and there exist (decoupled) $\delta$-visiting orders.

\begin{proof}
    We start with showing one direction of the equivalence statement. So,
    assume ${d_F(P,Q)\leq \delta}$. 
    Then, there exist values $v'_1\leq v'_2\leq\cdots \leq v'_t\in \mathbb{R}$ and $w'_1\leq w'_2\leq\cdots \leq w'_t\in \mathbb{R}$ such that 
    \begin{itemize}
        \item $d_F(P[v'_1, v'_2], Q[w'_1, w'_2])\leq \delta$ and $d_F(P[v'_{t-1}, v'_t], Q[w'_{t-1}, w'_t])\leq \delta$,
        \item $|P(v'_k)-Q(w'_k)|\leq \delta$,
        \item $P(v'_k)$ or $Q(w'_k)$ is a $\delta$-signature vertex for all $k=1, \dotso, t$, and
        \item all $\delta$-signature vertices of $P$ and $Q$ are contained in that sequence.
    \end{itemize}
    Set $v_k=v'_k$ and $w_k=w'_k$ for $k=1, 2, t-1, t$. For $k=3, \dotso, t-2$, consider the case that $Q(w'_k)$ is a $\delta$-signature vertex. Refer to the left of \cref{fig:Beweis_Main_Theo} for a visualization. Then, set $w_k=w'_k$. Let $Q(w'_a)$ be the $\delta$-signature vertex before $Q(w'_k)$ and $Q(w'_b)$ the one afterwards. If $w'_k$ is a local maximum (resp. minimum),
    then by the range property of a $\delta$-signature for all  $v'_a\leq v\leq v'_b$ holds that $P(v)\leq Q(w'_k)+\delta$ (resp. $P(v)\geq Q(w'_k)-\delta$).
    Therefore, one of the endpoints of the edge containing $P(v'_k)$ must lie in $[Q(w'_k)-\delta, Q(w'_k)+\delta]$ and we set $v_k$ to that point. A symmetric construction works for the case that $P(v'_k)$ is a $\delta$-signature vertex. 
    
     We need to show that $v_1\leq v_2\leq \cdots \leq v_t$ and $w_1\leq w_2\leq\cdots \leq w_t$. 
    Since we round each point $v'_k$ to an endpoint of the edge containing $v'_k$, the only time where it might get violated is if $v'_k$ and $v'_{k+1}$ lie on the same edge and $Q(w'_k)$ and $Q(w'_{k+1})$ are $\delta$-signature vertices (or with the role of $P$ and $Q$ changed). However, if $Q(w'_k)<Q(w'_{k+1})$ (resp. $Q(w'_k)>Q(w'_{k+1})$), then $P(v'_k)<P(v'_{k+1})$ (resp. $P(v'_k)>P(v'_{k+1})$) by the edge-length property of a signature. Hence, $v_k=\lfloor v'_k\rfloor$ and $v_{k+1}=\lceil v'_{k+1}\rceil$. So, $(v_1, w_1),\dotso, (v_t, w_t)$ is a coupled $\delta$-visiting order with the desired properties.
    
    To show the other direction, let $((v_1, w_1), \dotso, (v_t, w_t))$ be a coupled $\delta$-visiting order of~$P$ and $Q$ with $d_F(P[1, v_2], Q[1, w_2])\leq \delta$,  $d_F(P[v_{t-1}, n], Q[w_{t-1}, m])\leq~\delta$ and the properties of \cref{l:visiting_order_wproperties}.
    We show that \cref{l:VarLemma37} can be applied to $P[v_k, v_{k+1}]$ and $Q[w_k, w_{k+1}]$ for all $k=2, \dotso, t-2$. Then, it follows that $d_F(P[v_k, v_{k+1}],Q[w_k, w_{k+1}])\leq \delta$ and by \cref{o:concatanation}, it holds that $d_F(P, Q)\leq \delta$.

    Condition~(i), (ii) and (iv) of \cref{l:VarLemma37} are true by construction. So, it remains to prove that Conditions~(iii) and (v) are true. We show those properties for the case that $Q(w_k)$ is a $\delta$-signature vertex which is a local minimum of $Q$. The other cases follow analogously.
    \begin{enumerate}
        \item If $P(v_k), P(v_{k+1})$ are supported by a {$2\delta$-monotone} increasing $\delta$-signature edge, then $P(v_{k+1})$ or $Q(w_{k+1})$ is a $\delta$-signature vertex that is a local maximum, i.e., Condition~(v) holds.
        Let $P(v^*)$ denote the $\delta$-signature vertex before $P(v_{k+1})$ on $P$.
        \begin{enumerate}
            \item If $|P(v^*)-Q(w_k)|\leq \delta$, then $P[v_k, v_{k+1}]\geq Q(w_k)-\delta$. 
            \item Otherwise, $P(v^*)+\delta<Q(w_k)\leq Q[w_{k-1}, w_k]$. Refer to the right of \cref{fig:Beweis_Main_Theo}. Therefore, $Q(w_{k-1})$ is a $\delta$-signature and $P(v_{k-1})$ not. By the minimum edge length property and the $2\delta$-monotonicity, it follows that $P(s)\geq P(v_{k-1})-2\delta\geq (Q(w_{k-1})-\delta)-2\delta>Q(w_k)-\delta$ for all $v_{k-1}\leq s\leq v_{k+1}$.
        \end{enumerate}
        \item Otherwise, $P(v_k), P(v_{k+1})$ are supported by a {$2\delta$-monotone} decreasing $\delta$-signature edge. Then, $Q(w_{k+1})$ must be a $\delta$-signature vertex of $Q$ that is a global maximum on $Q[w_k, w_{k+1}]$ by Property (c) of \cref{l:visiting_order_wproperties}.
        Hence, Condition~(v) holds. Further, for all $v_k\leq v'\leq v_{k+1}$, it holds that ${P(v')\geq P(v_{k+1})-2\delta\geq Q(w_{k+1})-3\delta\geq Q(w_k)-\delta}$. 
    \end{enumerate}
    If $Q(w_{k+1})$ is a local maximum $\delta$-signature vertex, then $P[v_k, v_{k+1}]\leq Q(w_{k+1})+\delta$, resp. for $P$ and $Q$ switched. Therefore, Condition~(iii) is true.
    We have shown that all conditions of \cref{l:VarLemma37} are true for $P[v_k, v_{k+1}]$ and $Q[w_k, w_{k+1}]$, which completes the proof.
\end{proof}

\begin{figure}
    \centering
    \includegraphics[page=1]{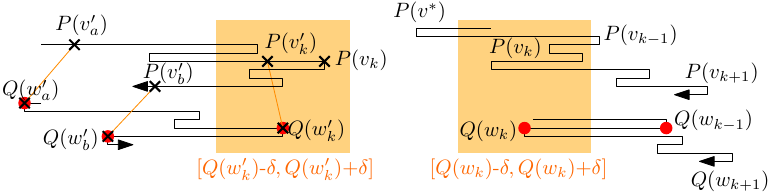}
    \caption{Visualization of parts of the proof of \cref{t:coupled visiting order}.}
    \label{fig:Beweis_Main_Theo}
\end{figure}

\Cref{l:visiting_order_wproperties} and \cref{t:coupled visiting order} yield the following: 
\begin{observation}\label{o:complexitySignature}
    If $d_F(P, Q)\leq \delta$, then the number of $\delta$-signature vertices of $P$ is at most two more than the complexity of $Q$.
\end{observation}

\Cref{t:coupled visiting order} implies the following corollary, which is Theorem~3.7 by Driemel, Krivošija and Sohler~\cite{DKS16}. We can even strengthen their statement to say $r_a=[P'(a)-3\delta, P'(a)+3\delta]$ instead of $[P'(a)-4\delta, P'(a)+4\delta]$ for $a=1$ or $a=s_P$. 
\Cref{fig:beginning-3delta possible} shows why we define $r_1=[P'(1)-3\delta, P'(1)+3\delta]$. 

\begin{corollary}[Theorem~3.7 in \cite{DKS16}]
    Let $P'=\langle P(i_1), \dotso, P(i_{s_P})\rangle$ be a $\delta$-signature of~$P$. Let $r_1=B(P(1),3\delta)$, $r_{s_P}=B(P(i_{s_P}),3\delta)$ and $r_k=B(P(i_k),\delta)$ for $k=2, \dotso, s_P-1$. Let $Q$ and $\widehat Q$ be time series such that $d_F(P,Q)\leq \delta$ and $\widehat{Q}$ is $Q$ after removing some vertex $Q(j')$ with $Q(j')\notin\bigcup_{1\leq k\leq l}r_k$. Then, it holds that $d_F(P, \widehat{Q})\leq \delta$.
\end{corollary}

\begin{proof}
    For simplicity of exposition, we duplicate the vertex $Q(j'-1)$ in $\widehat{Q}$. This way, the vertex numbering of the proceeding vertices is not affected by the removal of $Q(j')$.
    
    By \cref{t:coupled visiting order}, there exists a coupled $\delta$-visiting order 
    $V=((v_1, w_1), \dotso, (v_t, w_t))$ of $P$ and~$Q$ with ${d_F(P[1, v_2], Q[1, w_2])\leq \delta}$ and ${d_F(P[v_{t-1}, n], Q[w_{t-1}, m])\leq \delta}$. 

    Since $Q(j')$ does not lie inside the signature range $r_1$ of $P$ and $Q[1, w_2]\subset B(P[1, v_2], \delta)\subset B(P(1), 3\delta)$, it follows that $j'>w_2$.
    It similarly follows that $j'<w_{t-1}$. 
    
    If the $\delta$-signature $\widehat Q'$ of $\widehat Q$ is the same as the $\delta$-signature $Q'$ of $Q$, then $\widehat V=V$ is a coupled $\delta$-visiting order of $P$ and $\widehat Q$.  
    In this case, \cref{t:coupled visiting order} implies $d_F(P, \widehat Q)\leq \delta$.
    Otherwise, we show that there still exists a coupled $\delta$-visiting order of $P$ and $Q'$ and draw the same conclusion.
    
    Since the signature changed by the removal of $j'$, $Q(j')$ must be a $\delta$-signature vertex of $Q$ and as such there exists an $l$ such that $j'=w_l$. Now, let $Q(w_a)$ be the $\delta$-signature vertex before $Q(w_l)$ and $Q(w_b)$ the one after $Q(w_l)$ and let $j''\in [w_a, w_b]$ be such that $\widehat Q[w_a, w_b]\subset \overline{Q(w_a)Q(j'')}\cup \overline{Q(j'') Q(w_b)}$. We distinguish two cases: since the signature is an alternating sequence of minima and maxima, either (a) a neighbouring signature vertex must disappear from the signature or (b) a new signature vertex $j''$ must appear in lieu of $j'$.
    \begin{itemize}
        \item[a)] If $|Q(w_a)-Q(j'')|\leq 2\delta$ or $|Q(w_b)-Q(j'')|\leq 2\delta$, then one of $Q(w_a)$ or $Q(w_b)$ disappears together with $Q(j')$ from the signature.
        Hence, we can delete $(v_l, w_l)$ and $(v_a, w_a)$ (resp. $(v_b, w_b)$) from $V$ and obtain a coupled $\delta$-visiting order $\widehat V$ of $P$ and $\widehat Q$.
        

       
        \item [b)] Otherwise, it holds that $|Q(w_a)-Q(j'')|> 2\delta$ and $|Q(w_b)-Q(j'')|> 2\delta$ and that $\widehat Q'$ can be obtained from $Q'$ by replacing $Q(j')$ with $Q(j'')$. Assume without loss of generality\footnote{If one of these conditions is not fulfilled, we can mirror both curves at zero and/or reverse the direction of both curves.} that $Q(j'')$ is a maximum  and that $P(v_{l-1}) \geq P(v_{l+1})$. 
        We claim that $|P(v_l)-Q(j'')| \leq \delta$, where $P(v_l)$ by definition was the vertex matched to $Q(j')$ in the coupled visiting order. 
        
        The proof proceeds in two steps. First, note that $P(v_l)$ cannot be a $\delta$-signature vertex of $P$, otherwise $Q(j')$ would lie inside a $\delta$-signature range of $P$, which is ruled out by the conditions on $j'$ in the lemma statement. Hence, $P[v_{l-1}, v_{l+1}]$ is $2\delta$-monotone.
        We conclude that 
        \begin{align}
            P(v_l)\leq P(v_{l-1})+2\delta.\label{eq_1}
        \end{align}
        Indeed, if $P(v_{l})> P(v_{l-1})$ then this follows from the $2\delta$-monotonicity of $P[v_{l-1}, v_{l+1}]$ and if $P(v_{l})\leq P(v_{l-1})$ it follows trivially.
        
        Secondly, we argue that $P(v_{l-1})$ is not a $\delta$-signature vertex of $P$. Indeed, assume for the sake of contradiction that $P(v_{l-1})$ is a $\delta$-signature vertex. Then, it must be the maximum of $P[v_{l-1}, v_{l+1}]$ and in particular, $P(v_{l-1})\geq P(v_l)\geq Q(j')-\delta$. Further, as $Q(j')$ does not lie inside a signature range by the conditions in the lemma statement, it must hold that $P(v_{l-1})>Q(j')+\delta$. This leads to a contradiction as $P(v_{l-1})\leq Q(w_{l-1})+\delta\leq Q(j')+\delta$. Hence, $P(v_{l-1})$ cannot be a $\delta$-signature vertex of $P$ and therefore $Q(w_{l-1})$ must be a $\delta$-signature vertex of $Q$ by the definition of a coupled $\delta$-visiting order. It follows that $Q(w_{l-1})$ must be $Q(w_{a})$, the vertex preceding $Q(j'')$ on the signature $\widehat Q'$. Therefore,
        \begin{align}
            Q(w_{l+1})+2 = Q(w_{a})+2 \delta<Q(j'').\label{eq_2}
        \end{align} 
        Putting (\ref{eq_1}) and (\ref{eq_2}) together, we get that 
        \[-\delta\leq (Q(w_{l-1})+2\delta)-(P(v_{l-1})+2\delta) < Q(j'')-P(v_l) \leq Q(j')-P(v_l)\leq \delta.\]
        Hence, $|P(v_l)-Q(j'')|\leq \delta$ and we get a coupled $\delta$-visiting order $\widehat V$ of $P$ and $\widehat Q$ after replacing $(v_l, w_l)$ with $(v_l, j'')$ in $V$. 
    \end{itemize}
    Since $w_2< j'< w_{t-1}$, in any case $\widehat V$ fulfills the properties of \cref{t:coupled visiting order} for $P$ and~$\widehat Q$. Hence, $d_F(P, \widehat Q)\leq \delta$.
\end{proof}

\begin{figure}
    \centering
    \includegraphics[page=2, scale=0.9]{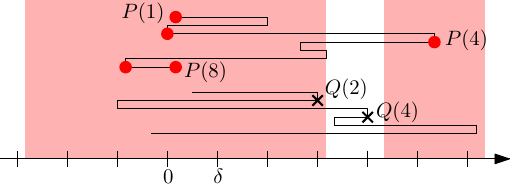}
    \caption{$P(1), P(3), P(4), P(7)$ and $P(8)$ are the $\delta$-signature vertices of $P$. If $\widehat Q$ is obtained by removing $Q(2)$ from $Q$, then $d_F(P, \widehat Q)> \delta$. Therefore, we define the (enlarged) signature-range $r_1=[P(1)-3\delta, P(1)+3\delta]$. If $\widehat Q$ is obtained from $Q$ by removing $Q(4)$, then it still holds that $d_F(P, \widehat Q)\leq \delta$.}
    \label{fig:beginning-3delta possible}
\end{figure}

\section{Fréchet Distance Oracle}\label{s:distnaceOracle}
In this section, we present a data structure to store a time series $P$ and to compute the \F between $P$ and a query time series $Q$. Note that the data structure can be built without knowing the complexity of the query.

\subparagraph*{The Data Structure.}
Let $P$ be a time series of complexity $n$. Then, we store the data structure of \cref{t:DS_signatures} and a layered range tree storing the points $(1, P(1)), (2, P(2)), \dotso, (n, P(n))$. 

\subsection{Decision Algorithm}
We begin with the decision variant of the problem.
Our query algorithm is based on \cref{t:coupled visiting order}. In previous works, the so-called free space diagram is used to solve the decision problem of the \Fm.  In contrast to this approach, we focus solely on \emph{super cells}, which arise from the grid superimposed by the $\delta$-signature edges onto the parametric space of the two time series. For example, the orange cell $[i_k, i_{k+1}]\times [j_l, j_{l+1}]$ in \cref{fig:Reachability_Graph} is a super cell. The computations within the super cell corresponding to the first $\delta$-signature edge of each time series, as well as for the last one (yellow cells in \cref{fig:Reachability_Graph}), are handled differently by Step~(B) and~(D). For all other super cells, Step~(C) computes the right exit point and top exit point. These exit points are defined as follows:

\begin{definition}
    Let $P$ and $Q$ be time series and $\langle P(i_1), \dotso, P(i_{s_P})\rangle$ and $\langle Q(j_1), \dotso, Q(i_{s_{Q}})\rangle$ be their extended $\delta$-signatures. Then for $k=1, \dotso, s_{P-1}$, $l=1, \dotso, s_{Q-1}$, we define the \emph{right (resp. top) exit point} $v_{k,l}\leq i_{k+1}$ (resp. {${w_{k,l}\leq j_{l+1}}$}) to be the smallest index such that there exist tuples $(v_1, w_1), \dotso, (v_r, w_r)$ with $v_r, w_r\in \mathbb{N}$ for $r=3, \dotso, t-2$, where
    \begin{enumerate}[(a)]
        \item $(v_r, w_r)=(v_{k,l}, j_l)$ (resp. $(v_r, w_r)=(i_k, w_{k,l})$),
        \item $v_1\leq\cdots \leq v_r$ and $w_1\leq\cdots \leq w_r$,
        \item $i_1, \dotso, i_k\in \{v_1, \dotso, v_r\}$ and $j_1, \dotso, j_l\in \{w_1, \dotso, w_r\}$,
        \item $|P(v_s)-Q(w_s)|\leq \delta$ for $s=1, \dotso, r$, and
        \item $d_F(P[1, v_2], Q[1, w_2])\leq \delta$.
    \end{enumerate}
    If no such index exists, we set $v_{k, l}=\infty$ (resp. $w_{k,l}=\infty$).
\end{definition}

\subparagraph*{The Query Algorithm.}
Let $Q$ denote the query time series of complexity $m$ and $\delta$ the distance parameter. Then, we check whether the $\delta$-signature of $P$ has at most $m$ vertices. If it has more, we stop and return $d_F(P,Q)>\delta$. Otherwise, the algorithm proceeds with the following four steps:
\begin{enumerate}[Step (A)]
    \item Compute the extended $\delta$-signatures $\langle P(i_1), \dotso, P(i_{s_P})\rangle$ of $P$ and $\langle Q(j_1), \dotso, Q(j_{s_Q})\rangle$ of~$Q$ and construct a layered range tree storing the points $(1, Q(1)), \dotso, (m, Q(m))$.
    
    \item Compute the first value $v_{1,2}\leq i_2$ such that $d_F(P[1, v_{1, 2}], Q[1, j_2])\leq \delta$ and the first value $w_{2, 1}\leq j_2$ such that $d_F(P[1, i_2], Q[1, w_{2,1}])\leq \delta$. If no such value exists, we set it to be infinity.

    \item Compute the exit points $v_{k,l}$ and $w_{k,l}$ by iterating over the super cells in a row-by-row order and within each row from left to right. 
    
    \item Compute the latest value $\widetilde v\geq i_{s_{P}-1}$ such that $d_F(P[\widetilde v, n], Q[j_{s_Q-1}, m])\leq \delta$ and the latest value $\widetilde w\geq j_{s_{q}-1}$ such that $d_F(P[i_{s_P-1}, n], Q[\widetilde w, m])\leq \delta$. If no such value exists, we set it to be zero. 
\end{enumerate}
If $v_{s_P-1, s_Q-1}\leq \widetilde v$ or $w_{s_P-1, s_Q-1}\leq \widetilde w$, return $d_F(P, Q)\leq \delta$. Otherwise, return $d_F(P, Q)>\delta$.
\ \\

In \cref{s:StBD} and \ref{s:StC}, we describe how Step~(B)-(D) can be computed. 

\begin{figure}
    \centering
    \includegraphics[page=2]{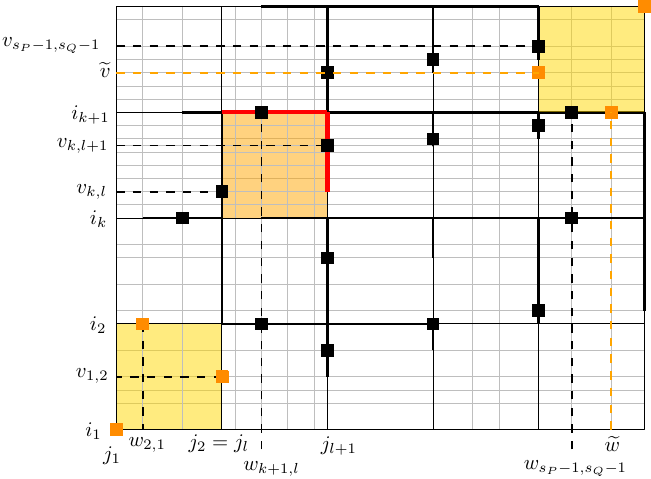}
    \caption{For the orange cell $[i_k, i_{k+1}]\times [j_l, j_{l+1}]$, we find the smallest index $v_{k, l-1}$ such that $(v_{k, l-1}, P(v_{k, l-1}))\in [v_{k, l}, i_{k+1}]\times B(Q(j_{l+1}),\delta)$ and the smallest index $w_{k+1, l}$ such that $(w_{k+1, l}, Q(w_{k+1, l}))\in [j_l, j_{l+1}]\times B(P(i_{k+1}), \delta)$. The thick lines mark the query intervals for the indices of the exit points. To compute $i_2, j_2, \widetilde v$ and $\widetilde w$, \cref{l: minimum first signature} is used.}
    \label{fig:Reachability_Graph}
\end{figure}

\begin{theorem}\label{t:decision}
    Given a time series $P$ of complexity $n$, we can construct a data structure that uses size and preprocessing time in $\mathcal{O}(n\log n)$ and supports the following type of queries. For a distance parameter $\delta>0$ and a time series $Q$ of complexity $m<n$, it can decide whether $d_F(P, Q)\leq \delta$ in time in $\mathcal{O}(m^2 \log n)$. If $s_P$ denotes the complexity of the $\delta$-signature of $P$, respectively $s_Q$ for $Q$, then the running time is also in $\mathcal{O}(s_P s_Q \log n+m\log n)$. 
\end{theorem}
\begin{proof}
    It takes time in $\mathcal{O}(n\log n)$ to compute the layered range tree of same size by \cref{t:successorProblem} and time in $\mathcal{O}(n\log n)$ to construct the data structure of \cref{t:DS_signatures} of size in $\mathcal{O}(n)$. With slight modifications, it can answer also queries as in the beginning and Step~(A) of the query algorithm.
    Using the definition of \cref{l:hierachicalSignature}, this data structure stores the vertices of $P$ exactly once together with some token separators such that the vertices after the $i$-th separator belong to $P_i$. 
    By \cref{l:hierachicalSignature}, we can store a pointer from the $i$-th separator to the interval $[\delta_i, \delta_{i+1})$ such that the $\delta$-signature, for any $\delta\in [\delta_i, \delta_{i+1})$, contains all vertices after this separator. 
    To check whether the $\delta$-signature of $P$ has at most $m$ vertices takes $\mathcal{O}(m)$ time. If it has more than $m$ vertices, the total running time is in $\mathcal{O}(m)$. Otherwise, it holds that $s_P\leq m$. Step~(A) takes $\mathcal{O}(m\log m)$ time by \cref{t:DS_signatures} and \ref{t:successorProblem}.
    To compute $v_{1, 2}$, $w_{2, 1}$, $\widetilde v$ and $\widetilde w$ takes $\mathcal{O}(m \log n)$ time by \cref{l: minimum first signature}. Each iteration of Step~(C) takes $\mathcal{O}(\log n)$ time by \cref{l:step C}. There are $\mathcal{O}(s_Ps_Q)$ steps. Since $s_Q, s_P\leq m$, the total running time is in $\mathcal{O}(m^2 \log n)$.
    
    It remains to prove the correctness. If the $\delta$-signature of $P$ has more than $m$ vertices, $d_F(P,Q)>\delta$ by \cref{o:complexitySignature}. Assume that $v_{s_P-1, s_Q-1}\leq \widetilde v$. We add vertices at $P(v_{1,2})$, $Q(w_{2, 1})$, $P(\widetilde v)$ and $Q(\widetilde w)$ if there were none before. Then, by the definition of left exit points, there exist tuples such that $(v_1, w_1), \dotso, (v_{t-2}, w_{t-2}), (\widetilde v, j_{s_Q-1}), (n,m)$ is a coupled $\delta$-visiting order with $d_F(P[1, v_2], Q[1, w_2])\leq \delta$. Further, by time Step~(D), it holds $d_F(P[\widetilde v, n], Q[j_{s_Q-1}, m])\leq \delta$. Hence, $d_F(P, Q)\leq \delta$ by \cref{t:coupled visiting order}.
    In the case that $d_F(P,Q)\leq \delta$, there exist a coupled $\delta$-visiting order $(v_1, w_1), \dotso, (v_t, w_t)$ of $P$ and $Q$ such that $d_F(P[1, v_2], Q[1, w_2])\leq \delta$ and $d_F(P[v_{t-1}, n], Q[w_{t-1}, m])\leq \delta$ by \cref{t:coupled visiting order}. 
    Hence, $v_{s_P-1, s_Q-1}\leq v_{t-1}\leq \widetilde v$ or $w_{s_P-1, s_Q-1}\leq w_{t-1}\leq \widetilde w$. So, the algorithm returns $d_F(P,Q)\leq \delta$ and is correct.
\end{proof}

\subsubsection{Step~(B) and Step~(D)}\label{s:StBD}
Step~(D) can be done in the same way as Step~(B), after reversing the order of the vertices in both time series. Therefore, we only describe Step~(B) here.
\Cref{Alg:first_edge} computes the value~$w_{2,1}$. With exchanged roles of $P$ and $Q$, it computes~$v_{1,2}$. An example of a matching produced by this algorithm is pictured in \cref{fig:MatchingUntilMinimum}.

\begin{lemma}\label{l: minimum first signature}
    Given the layered range trees for $P$ and $Q$, it is possible to compute $w_{2,1}$ in $\mathcal{O}(\min\{m, n\} (\log n+\log m))$ time.
\end{lemma}

\begin{algorithm}
    \caption{Computation of $w_{2,1}$\label{Alg:first_edge}}
    Find minimum $w_{2,1}\leq j_2$ such that $P[1, i_2]\subset B(Q(w_{2,1}), \delta)$\;
    \If{there does not exist one}
        {\textbf{return} $w_{2,1}=\infty$\;}
    Set $s=i_2$ and $t=w_{2,1}$\;
    \While{$s\neq 1$ or $t\neq 1$}{
        Find minimum $s'\leq s$ such that $Q[1, t]\subset B(P(s'), \delta)$\;
        \If{$s'=s$}
        {\textbf{return} $w_{2,1}=\infty$\;}
        Set $s=s'$\;
        Find minimum $t'\leq t$ such that $P[1, s]\subset B(Q(t'), \delta)$\;
        \If{$t'=t$}
        {\textbf{return} $w_{2,1}=\infty$\;}
        Set $t=t'$\;
    }
    \textbf{return} $w_{2,1}=w_{2,1}$\;
\end{algorithm}
\begin{proof}
    We show that if \cref{Alg:first_edge} returns $w_{2,1}<\infty$, then it is the minimum $j\leq j_2$ such that $d_F(P[1, i_2], Q[1, j])\leq \delta$ and otherwise there does not exist one.
    If $w_{2,1}<\infty$, we found a traversal: Every time we reassign $s$ (resp. $t$), the part $P[s', s]$ (resp. $Q[t', t]$) is matched to $Q(t)$ (resp. $P(s)$). By construction, it holds that their distance is at most $\delta$. Hence, $d_F(P[1, i_2], Q[1, w_{2,1}])\leq \delta$. See \cref{fig:MatchingUntilMinimum} for an example.

    Let $s^*\leq i_2$ be such that $P[1, i_2]\subset \overline{P(s^*) P(i_2)}$. Further, assume that there exist $t_1\leq j_2$ and $t_2\leq j_2$ such that $|P(i_2)-Q(t_1)|\leq \delta$ and $|P(s^*)-Q(t_2)|\leq \delta$. Then, there must exist a point $t'$ on $Q$ before $j_2$ such that $P[1, i_2]\subset B(Q(t'), \delta)$, because $P[1, i_2]$ is contained in a $2\delta$-range by the definition of $\delta$-signature. So, if the algorithm returns $w_{2,1}=\infty$ in Line~3, then there exist an $s\leq i_2$ such that $|P(s)-Q(t')|>\delta$ for all $t'\leq j_2$. Hence, for all $j\leq j_2$ it holds $d_F(P[1, i_2], Q[1, j])>\delta$. This also shows that if $w_{2,1}<\infty$, it is the minimum $t$ such that $d_F(P[1, i_2], Q[1, t])\leq \delta$. 

    In the case that the algorithm returns $w_{2,1}=\infty$ in Line~8 or~12, it holds that $P(s')\neq P(s)$ for all $s'<s$, where $s$ is the value at the time it stops. Therefore, $P(s)$ is the maximum or minimum of $P[1, s]$, respectively $Q(t)$ of $Q[1, t]$. 
    Let $s^*\leq s$ be such that $P[1, s]\subset \overline{P(s^*) P(s)}$. For all $t'< t$, it holds that $|P(s)-Q(t)|\leq \delta$. Therefore, $|P(s^*)-Q(t')|>\delta$. In the same way, it follows that there exists a point $t^*<t$ such that $|Q(t^*)-P(s')|>\delta$ for all $s'<s$. Hence, $d_F(P, Q)>\delta$, because $P(s^*)$ is not matchable to anything before $Q(t)$ and $Q(t^*)$ not to anything before $P(s)$.

    It remains to prove the running time.
    Assume there exist two values $s_1<s_2$ on the same edge of $P$ that were both $s'$ at some point during the algorithm and are no vertices. Without loss of generality, let $P(s_1)<P(s_2)$. Then, $P(s_1)=\min(P[1, s_1])$ and $P(s_2)=\min(P[1, s_2])$. Let $t_1$ be the value $t'$ attained after $s'=s_1$. Then, $Q(t_1)=P(s_1)+\delta$ as otherwise the algorithm would have stopped the iteration after. Further, since $P(s_2)=\min(P[1, s_2])$, it must hold that $P(s_2)=Q(t_1)+\delta=P(s_1)+2\delta$. This is a contradiction to $P[1, i_2]$ is contained in a $2\delta$-range and $P(s_1)$, $P(s_2)$ are no vertices. The same arguments work for $Q(t_1)$ and $Q(t_2)$. Hence, there are at most $\mathcal{O}(\min\{m, n\})$ steps until $s=1$ and $t=1$ or the algorithm stops. Each iteration takes  $\mathcal{O}(\log n+ \log m)$ time by \cref{t:successorProblem}. So, the total running time is in $\mathcal{O}(\min\{m, n\}(\log n +\log m))$.
\end{proof}

\begin{figure}
    \centering
    \includegraphics{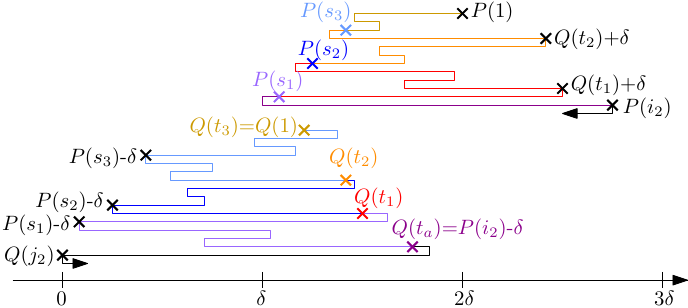}
    \caption{Example of a matching computed by \cref{Alg:first_edge}. The parts matched together are of the same color.}
    \label{fig:MatchingUntilMinimum}
\end{figure}


\subsubsection{Step~(C): Computation of the Exit Points}\label{s:StC}
In Step~(C), we compute iteratively the right and top exit points. For $k=1, \dotso, s_Q-1$, we set $v_{k,1}=\infty$ and $w_{1, k}=\infty$ since there do not exist exit points for those values by Property~(c) and (e) of the definition.

Given the values $v_{k, l}$ and $w_{k, l}$, we compute the values $v_{k, l+1}$ and $w_{k+1, l}$ in the following way: 
If $v_{k, l}\neq \infty$, set $w_{\min}=j_l$. Otherwise, set $w_{\min}=w_{k,l}$. Then, using the layered range tree, we can compute the minimum index $w_{k+1,l}\leq j_{l+1}$ such that $(w_{k+1,l}, Q(w_{k+1,l}))\in [w_{\min}, j_{l+1}]\times B(P(i_{k+1}), \delta)$. If there does not exist such a point, we set $w_{k+1, l}=\infty$. Symmetrically, if $w_{k, l}\neq \infty$, set $v_{\min}=i_k$. Otherwise, set $v_{\min}=v_{k,l}$. Then, we compute the minimum value $v_{k,l+1}$ such that $(v_{k,l+1}, Q(v_{k,l+1}))\in [v_{\min}, i_{k+1}]\times B(Q(j_{l+1}), \delta)$. If there does not exist such a point, we set $v_{k, l+1}=\infty$. 
The orange super cell in \cref{fig:Reachability_Graph} visualizes this construction. Here, $w_{k, l}=\infty$ and the red bars show the query intervals for the indices.

\begin{lemma}\label{l:step C}
    Given the indices $v_{k,l}$ and $w_{k,l}$, the indices $v_{k, l+1}$ and $w_{k+1, l}$ can be computed in $\mathcal{O}(\log n)$ time.
\end{lemma}
\begin{proof}
    We use the algorithm above to compute $v_{k, l+1}$ and $w_{k+1, l}$. Its running time is in $\mathcal{O}(\log n)$ by \cref{t:successorProblem}.
    Let $v^*$ be the value computed by the algorithm for $v_{k, l+1}$.
    If $v^*\neq\infty$, then $v_{k,l}\neq \infty$ or $w_{k,l}\neq \infty$ and one of them is an exit point. Therefore, there exist tuples $(v_1, w_1), \dotso, (v_r, w_r), (v^*, j_{l+1})$ fulfilling the properties (a)-(e) of the definition of exit points. It remains to prove that it is the minimum.
    Let $v_{k, l+1}$ be the left exit point. Then, by definition of the exit points it holds that $v_{k,l}\neq \infty$ or $w_{k,l}\neq \infty$. 
    Further, it must hold that $|P(v^*)-Q(j_{l+1})|\leq \delta$ and $v_{k, l+1}\leq i_{k+1}$ and $w_{\min}\leq v^*$. Therefore, $v^*=v_{k, l+1}$. Symmetrically, it follows that the construction of $w_{k+1,l}$ is correct.
\end{proof}

\subsection{Critical values for computation}
To compute the exact \Fm, we search over a set of critical values for $\delta$ using the decision algorithm in each step. In particular, we define the set of values of $\delta$ for which the outcome of algorithm of \cref{t:decision} might change. 

For curves in~1D, this set consists of (i) the distances between a vertex of $P$ and one of $Q$, i.e., $|P(i)-Q(j)|$ for $i=1, \dotso, n$, $j=1, \dotso, m$, and (ii) the $\delta$ values of \cref{l:hierachicalSignature} where the signature of $P$ or $Q$ changes. 
For any two values in the interior of an atomic interval of this set of critical values, the decision algorithm returns the same answer.
Hence, the continuous \F between $P$ and $Q$ must be one of the critical values. 
In the following, we assume without loss of generality that $m\leq n$.

Our search over these values is carried out in several (but a constant number) of stages. In the end, we combine the results of the different stages by taking the minimum over all valid answers.

We first discuss how to search over the values of type (i).
We can extract an (increasing-order) sorted list $A$ of the $n$ vertices of $P$ from the data structure in linear time. In addition, we compute a sorted list $B$ of the $m$ vertices of $Q$. Consider the implicit $n\times m$ matrices $M_1$ and $M_2$, where the entry $M_1(i, j)$ is the value $\max(0,a_i-b_j)$, and the entry $M_2(i, j)$ is the value $\max(0,b_j-a_i)$. Observe that both matrices consist of entries that are sorted in each row and column.
Now, Fredrickson and Johnson \cite{FJ84} showed that we can find the $k$-th smallest item in $M_1$ (resp. $M_2$) in $\mathcal{O}(m \log (2n/m))$ time. Hence, we can perform an implicit binary search over all values in $M_1$ and $M_2$ with total running time in $\mathcal{O}(\log n (m\log (2n/m)+m^2\log n))$. 

For the critical values of type (ii), the $\mathcal{O}(n)$values of $\delta$, where the $\delta$-signature changes, can be extracted in sorted order from the data structures of $P$ and $Q$. Once we have these values stored in sorted order, we can perform a binary search over them using $\mathcal{O}(\log n)$ calls to the decision algorithm.

Overall, using \cref{t:decision} we obtain the following result.
\begin{theorem}\label{t:distance oracle}
    Given a time series $P$ of complexity $n$, we can construct a data structure that returns $d_F(P, Q)$ for a time series $Q$ of complexity $m\leq n$. It has size and preprocessing time in $\mathcal{O}(n\log n)$ and its query time lies in $\mathcal{O}(m^2\log^2 n)$.
\end{theorem}

This theorem leads directly to a new algorithm to compute the \F in 1D:
\begin{corollary}
    Let $P$ be a time series of complexity $n$ and $Q$ of complexity $n^{\alpha}$ with $\alpha\in[0,1]$. Then, it is possible to compute the \F between them in $\mathcal{O}(n^{2\alpha}\log^2 n+n\log n)$ time.
\end{corollary}

\section{Fr\'echet Distance Oracle with Subcurve Queries}
The data structure of \cref{s:distnaceOracle} can be modified such that it can also answer queries for subcurves of $P$ using the same query algorithm. The layered range tree storing the points $(1, P(1)), (2, P(2)), \dotso, (n, P(n))$ can also be used to answer the needed queries for subcurves of $P$. Therefore, the only additional requirement is the ability to compute the $\delta$-signature for subcurves of $P$. 

The next lemma, shows how the $\delta$-signatures for subcurves of $P$ can be computed when the $\delta$-signature of $P$ is given.

\begin{lemma}\label{l:SubcurveSignature}
    Let $P$ be a time series of complexity $n$ and $\langle P(i_1), \dotso, P(i_{s_P})\rangle$ its $\delta$-signature. Then, for values $1\leq s\leq t\leq n$, it is possible to compute the extended $\delta$-signature $P'$ of $P[s, t]$ in $\mathcal{O}((l-k)+\log n)$ time, when given the layered range tree storing the points $(1, P(1)), \dotso, (n, P(n))$ and the values $i_{k-1}< s\leq i_k\leq i_{k+1}\leq\cdots\leq i_l\leq t<i_{l+1}$.
\end{lemma}
\begin{proof}
    If there does not exist an $i_k$ such that $s\leq i_k\leq t$, then compute the global minimum and maximum of $P[s, t]$. Let the indices of those two points be $i^*\leq i^{**}$. Further, let $i'\in [s, i^*]$ such that $P[s, i^*]\subset \overline{P(i')P(i^*)}$ and $i''\in [i^{**}, t]$ such that $P[i^{**}, t]\subset \overline{P(i'')P(i^{**})}$. Those vertices can be found in $\mathcal{O}(\log n)$ time using the layered range tree by \cref{t:successorProblem}. The vertices of the extended $\delta$-signature of $P[s, t]$ are $P(s), P(i^*), P(i^{**}), P(t)$ and if $P[s, t]\notin B(\overline{P(s) P(t)}, \delta)$, then in addition $P(i')$ if $|P(s)-P(i^*)|>\delta$ and $P(i'')$ if $|P(i^{**})-P(t)|>\delta$. For this curve, all conditions of the definition of an extended $\delta$-signature are fulfilled.

    Otherwise, we compute $i^*\in [s, i_k]$ such that $P[s, i_k]\subset \overline{P(i^*) P(i_k)}$, and $i'\in [s, i^*]$ such that $P[s, i^*]\in \overline{P(i') P(i^*)}$, and $i^{**}\in [i_l, t]$ such that $P[i_l, t]\subset \overline{P(i_l) P(i^{**})}$, and $i''\in [i^{**}, t]$ such that $P[i^{**}, t]\in \overline{P(i^{**})P(i'')}$. Those points are well defined by construction and the definition of a $\delta$-signature and we can compute those points using the layered range tree in $\mathcal{O}(\log n)$ time. 
    The vertices of the extended $\delta$-signature of $P[s, t]$ are 
    $P(s)$, $P(t)$, $P(i_k), P(i_{k+1}), \dotso, P(i_l)$ and in addition
    \begin{itemize}
        \item[i)] if $|P(s)-P(i_k)|\leq \delta$ and $|P(i^*)-P(i_k)|> 2\delta$, then $P(i')$ and $P(i^*)$,
        \item[ii)] if $|P(s)-P(i_k)|> \delta$ and $|P(s)-P(i^*)|\leq \delta$, then $P(i^*)$,
        \item[iii)] if $|P(s)-P(i_k)|>\delta$ and $|P(s)-P(i^*)|>\delta$, then $P(i')$ and $P(i^*)$,
        \item[iv)] if $|P(t)-P(i_l)|\leq \delta$ and $|P(i^{**})-P(i_l)|> 2\delta$, then $P(i'')$ and $P(i^{**})$,
        \item[v)] if $|P(t)-P(i_l)|> \delta$ and $|P(t)-P(i^{**})|\leq \delta$, then $P(i^{**})$, and
        \item[vi)] if $|P(t)-P(i_l)|>\delta$ and $|P(t)-P(i^{**})|>\delta$, then $P(i'')$ and $P(i^{**})$.
    \end{itemize}
    By construction, $P(i^*)$ is the global minimum (resp. maximum) of $P[s, i_k]$ if $P[i_{k-1}, i_k]$ is $2\delta$-monotone increasing (resp. decreasing). Therefore, $P[s, i^*]$ is contained in a $2\delta$-range. So, the properties of the definition of a {$\delta$-signature} follow by construction and because $\langle P(i_1), \dotso, P(i_{s_P})\rangle$ is the $\delta$-signature of~$P$. All those points can be found in $\mathcal{O}((l-k)+\log n)$ time.
\end{proof}

The data structure of \cref{t:DS_signatures} can compute the $\delta$-signatures of $P$. Using this and the lemma above, it is possible to construct a data structure that can compute the extended $\delta$-signatures of subcurves of $P$.

\begin{lemma}\label{l:SubcurveDS}
    Given a time series $P$ of complexity $n$, there exists a data structure using size and preprocessing time in $\mathcal{O}(n\log n)$ that can compute for any values $1\leq s\leq t\leq n$ and any $\delta\geq 0$, the extended $\delta$-signature $P'$ of $P[s, t]$ in time in $\mathcal{O}(m+\log n)$, where $m$ denotes the complexity of $P'$. Further, for $1\leq s\leq t\leq n$ and a number $m\in \mathbb{N}$, it can decide in $\mathcal{O}(\log n)$ time whether the complexity of the $\delta$-signature $P'$ is in $\mathcal{O}(m)$. 
\end{lemma}
\begin{proof}
    Using the data structure of \cref{t:DS_signatures} and \cref{l:hierachicalSignature}, we can compute in $\mathcal{O}(n\log n)$ time the values $\delta(k)$ for $k=1, \dotso, n$ such that $P(k)$ is a vertex of the $\delta$-signature of $P$ for $\delta<\delta(k)$. Then, we store the points $(k, \delta(k))$ in a layered range tree. The points $(k, P(k))$ for $k=1\dotso, n$ are stored in a second layered range tree. The size of both data structures is in $\mathcal{O}(n\log n)$ by \cref{t:successorProblem}.
    Now, given $1\leq s\leq t\leq n$ and $\delta$, we can compute the indices $i_k\in [s, t]$ such that $\delta(k)>\delta$ in $\mathcal{O}(m'+\log n)$ time using the first layered range tree, where $m'$ is the number of the indices with those properties. By construction, those are the indices $i_k, i_{k+1}, \dotso, i_l$ of \cref{l:SubcurveSignature}. Hence, using \cref{l:SubcurveSignature}, we can compute the extended $\delta$-signature of $P[s, t]$ in $\mathcal{O}(m+\log n)$ time.

    To check whether the complexity of the $\delta$-signature of $P[s, t]$ is in $\mathcal{O}(m)$, we compute the number of point stored in the first layered range tree that are contained in $[s, t]\times [\delta, \infty]$. This can be done in $\mathcal{O}(\log n)$ time by \cref{t:successorProblem}. Further, by \cref{l:SubcurveSignature}, it holds that if this number is in $\mathcal{O}(m)$, then the complexity is in $\mathcal{O}(m)$. Otherwise the complexity is higher.
\end{proof}

Using the data structure of \cref{l:SubcurveDS}, we get the following theorem.

\begin{theorem}
    Given a time series $P$ of complexity $n$, we can construct a data structure that uses size and preprocessing time in $\mathcal{O}(n\log n)$ and supports the following type of queries. For a time series $Q$ of complexity $m\leq n$, two values $1\leq s \leq t\leq n$ and a value $\delta \geq 0$, it can decide whether $d_F(P[s, t], Q)\leq \delta$ in $\mathcal{O}(m^2\log n)$ time. Further, it can compute the continuous \F between $P[s, t]$ and $Q$ in $\mathcal{O}(m^2\log^2 n)$ time.
\end{theorem}
\begin{proof}
    We store the data structure of \cref{l:SubcurveDS} and the layered range tree storing the points $(1, P(1)), (2, P(2)), \dotso, (n, P(n))$ using size and preprocessing time in $\mathcal{O}(n\log n)$ by \cref{t:successorProblem} and \cref{l:SubcurveDS}. Then, we can use the same query algorithm as it was used in \cref{t:decision} and \ref{t:distance oracle} and the query time follows by those theorems.
\end{proof}

\section{Conclusions}
In this paper, we presented an exact \F oracle in 1D. It stores a curve $P$ using space and query time in $\mathcal{O}(n\log n)$ and computes the \F to a query curve $Q$ of complexity $m$ in $\mathcal{O}(m^2\log^2 n)$ time. Moreover, this leads directly to a faster algorithm for computing the \F in 1D in the imbalanced case. This shows that the 2D lower bound by Bringmann~\cite{bringmann2014walking} does not fully translate to the 1D setting. 


\bibliography{Literature.bib}

\appendix

\section{Proof of \cref{l:VarLemma37}}
Bringmann, Driemel, Nusser and Psarros \cite{BDNP21} showed \cref{l:VarLemma37} for the case that $Q(1)$ is a global minimum (resp. maximum) and $Q(m)$ is a global maximum (resp. minimum). 
They gave a traversal of $P,Q$ to show that $d_F(P,Q)\leq \delta$. This is done as follows.
Assume that $P, Q$ are $2\delta$-monotone increasing and $Q(1)$ is a global minimum of $Q$, as the other cases are symmetric. The traversal tries to maintain the following two invariants. Here, the position during the traversal is denoted with $(s, t)\in [1,n]\times [1, m]$.
\begin{enumerate}
    \item[(1)] $P$ and $Q$ are in position $(s, t)$ such that $P(s)=Q(t)+\delta$.
    \item[(2)] The suffix of $Q$ is strictly greater than the current value $Q(t)$, i.e., $Q(t')>Q(t)$ for all $t'>t$.
\end{enumerate}
In the beginning of the traversals, we traverse $P$ until it first reaches $Q(1)+\delta$, while in~$Q$ we stay in $Q(1)$. When Invariant (1) would be violated, we traverse $P$ while staying in $Q(t)$ on~$Q$ until the next time we reach a position $s'$ on $P$ such that $P(s')=P(s)$. When Invariant~(2) would be violated, we continue traversing $Q$ while staying in $P(s)$ on $P$ until we reach the largest position $t'>t$ such that $Q(t')=Q(t)$. Whenever we reached the end of one time series, we stay there and finish traversing the other time series till the end.

\begin{figure}
    \centering
    \includegraphics{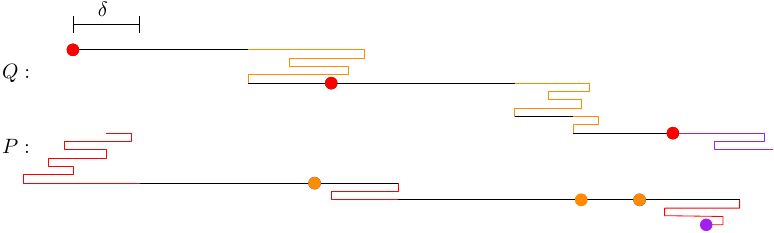}
    \caption{Visualization of the algorithm realizing the traversal in \cref{l:VarLemma37}. The red stems from the case when we restore Invariant (1) and the orange when we restore Invariant (2).}
    \label{fig:Alg Lemma~37}
\end{figure}

\begin{proof}[Proof of \cref{l:VarLemma37}]
    We prove the lemma for the case that $P, Q$ are $2\delta$-monotone increasing and $Q(1)$ is a global minimum of $Q$. The other cases follow symmetric. 
    It holds that $Q(n)\leq P(n)+\delta$ by Condition (ii). Further, by Condition (iii) and (v) it follows that $Q(t)\leq P(n)+\delta$ for all $t\in [1, m]$. Similarly, it follows that $P(s)\leq Q(m)+\delta$ for all $s\in [1, n]$.  
    
    In the beginning, we stay in $Q$ at $Q(1)$ while traversing $P$ until first reaching $Q(1)+\delta$. By Condition (ii), we have $P(1)\leq Q(1)+\delta$ and hence by Condition (iii) all points have distance at most $\delta$ to $Q(1)$. If we reach $P(n)$ before reaching $Q(1)+\delta$, we traverse $Q$ until the end while staying at $P(n)$. It holds that $P(n)-\delta \leq Q(1)\leq Q(t)\leq P(n)+\delta$ for all $t\in [1, m]$ by the observation above. Hence, we found a feasible traversal.
    
    From now on, we traverse $P$ and $Q$ with the same speed in image space, unless one of the two invariants would be violated by continuing the traversal. Here, it holds that $|P(s)-Q(t)|\leq \delta$ by the proof of Lemma 37 in \cite{BDNP21} for all $(s, t)$  until we reached the end of one time series.

    If we reach $P(n)$ while traversing $P$ and $Q$ with same speed in image space, it holds $P(n)-\delta\leq Q(t)\leq Q(t')\leq P(n)+\delta$ for all $t\leq t'\leq m$ by Invariant (1) and (2) and the observation above.
    If we reached $P(n)$ while restoring Invariant (1), $P(n)\leq Q(t)+\delta$ and $Q(t)\leq Q(t')$ for all $t'\geq t$ by Invariant (i), (ii). Hence, $P(n)-\delta\leq Q(t)\leq Q(t')\leq P(n)+\delta$ by the observation above.
    Assume we reached $Q(m)$ while traversing $P$ and $Q$ with same speed in image space or restoring Invariant (2). Then, $Q(m)-\delta=P(s)-2\delta\leq P(s')\leq Q(m)+\delta$ for all $s\leq s'\leq n$ by Invariant (1) and Condition (i).
\end{proof}

\end{document}